\documentclass[preprint,12pt]{elsarticle}

\usepackage{epsfig}
\usepackage{amsthm}
\usepackage{amsmath} 
\usepackage{amssymb}
\usepackage{amsfonts}
\usepackage{mathpazo}
\usepackage{amssymb}

\usepackage{hhline}
\usepackage{float}

\usepackage{algorithm}

\usepackage{algpseudocode}

\usepackage{array}
\usepackage[normalem]{ulem}
\usepackage{hyperref}
\usepackage{times}

\usepackage{refcount}

\usepackage[textsize=scriptsize]{todonotes}

\theoremstyle{plain}

\newtheorem{theorem}{Theorem}
\newtheorem{proposition}[theorem]{Proposition}
\newtheorem {corollary}[theorem]{Corollary}
\newtheorem{definition}[theorem]{Definition}
\newtheorem{lemma}[theorem]{Lemma}

\newtheorem{example}[theorem]{Example}

\newcommand{\class}[1]{\mathcal{#1}}

\newcommand{\vij}{v_{i,j}}
\newcommand{\vv}[2]{v_{#1,#2}}

\newcommand{\gsr}{\textit{GSR}}
\newcommand{\GSR}{\textit{GSR}}

\newcommand{\ftg}{\textit{FTG}}
\newcommand{\FTG}{\textit{FTG}}

\newcommand{\lnext}{{\sf Lnext}}
\newcommand{\Lnext}{{\sf Lnext}}

\newcommand{\cover}{\textit{cover}}

\newcommand{\cfl}{\textit{CFL}}
\newcommand{\CFL}{\textit{CFL}}

\newcommand{\Duval}{{\sf Duval}}

\newcommand{\ctf}{{\sf cover\_to\_FTG}}

\begin{document}

\sloppy

\journal{Discrete Mathematics}

\begin{frontmatter}

\title{An Efficient Generalized Shift-Rule for the Prefer-Max De Bruijn Sequence}

\author[tec]{Gal Amram}
\ead{galamram@technion.ac.il}

\author[bgu]{Amir Rubin}
\ead{amirrub@cs.bgu.ac.il}

\address[tec]{Department of Electrical Engineering, Technion-Israel Institute of
Technology} 
\address[bgu]{Department of Computer Science, Ben-Gurion University of The Negev}

\begin{abstract}
  One of the fundamental ways to construct De Bruijn sequences is by using a shift-rule. A shift-rule receives a word as an argument and computes the symbol that appears after it in the sequence. An optimal shift-rule for an $(n,k)$-De Bruijn sequence runs in time $O(n)$. We propose an extended notion we name a generalized-shift-rule, which receives a word, $w$, and an integer, $c$, and outputs the $c$ symbols that comes after $w$. An optimal generalized-shift-rule for an $(n,k)$-De Bruijn sequence runs in time $O(n+c)$. We show that, unlike in the case of a shift-rule, a time optimal generalized-shift-rule allows to construct the entire sequence efficiently. We provide a time optimal generalized-shift-rule for the well-known prefer-max and prefer-min De Bruijn sequences.

\end{abstract}


\end{frontmatter}

\section{Introduction}

De Bruijn sequences were rediscovered many times over the years, starting from 1894 by Flye-Sainte Marie~\cite{Fly1894}, and finally by De Bruijn himself in 1946~\cite{de1946}. For two positive non-zero integers, $k$ and $n$, an $(n,k)$-De Bruijn ($(n,k)$-DB, for abbreviation) sequence is a cyclic sequence over the alphabet $\{0,\dots,k-1\}$ in which every word of length $n$ over $\{0,\dots,k-1\}$ appears exactly once as a subword. 
It is cyclic in the sense that some words are generated by  concatenating the suffix of length $m<n$ of the sequence, with its prefix of length $n-m$.

A construction for a family of $(n,k)$-DB sequences is an algorithm that receives the two arguments, $n$ and $k$ (occasionally, $k$ is fixed and only $n$ is given as argument), and outputs an $(n,k)$-DB sequence. Obviously, a trivial time lower bound for a construction is $\Omega(k^n)$, as this is the exact length of an $(n,k)$-DB sequence. Many constructions for a variety of families of De Bruijn sequences are known, (for example, \cite{alhakim,amram2019efficient,grandmama,etzion,FM,gabric,newalgo,martin,simple,sawada17}) and some of them are also time optimal.

A specifically famous family of $(n,k)$-DB sequences is the prefer-max family~\cite{ford57,martin}, which is constructed by the well-known ``granddaddy" greedy algorithm~\cite{martin} (see also~\cite[Section~7.2.1.1]{knuth2011}). The algorithm constructs the sequence symbol by symbol, where at each step the maximal value is added to the initial segment constructed so far, so that the new suffix of length $n$ does not appear elsewhere. A symmetric approach produces the prefer-min DB sequence. Besides this highly inefficient algorithm,
many other constructions for the prefer-max and prefer-min sequences have been proposed in the literature. A classic result by Fredricksen and Kessler~\cite{FK}, and Fredricksen and Maiorana~\cite{FM} shows that the prefer-max sequence is in fact a concatenation of certain words, a result we use in this work. This block construction was later proved to be time optimal in~\cite{ruskey92}. Another efficient block concatenation construction was suggested in~\cite{ralston81}.

A common and important way of generating DB sequences is by using a shift-rule (also named a shift-register). A shift-rule for an $(n,k)$-DB sequence receives a word of length $n$, $w$, as an input, and outputs the symbol that follows $w$ at the sequence. Here, $n$ and $k$ are parameters of the algorithm. Obviously, a shift-rule must run in $\Omega(n)$ time since it must read every symbol in its input to produce the correct output. Shift-rules are important since, unlike block constructions,  they can be applied on words that appear at the middle of the sequence.

Several efficient shift-rules for DB sequences are known for $k=2$ (see~\cite{simple} for a comprehensive list). However, only recently efficient shift-rules were discovered for non-binary sequences. Sawada et al.~\cite{sawada17} introduced a new family of DB sequences and provided a linear time shift-rule for these sequences. Amram et al.~\cite{amram2019efficient} introduced an efficient shift-rule for the famous prefer-max and prefer-min DB sequences.

We note that, generally, a construction for a DB sequence provides an exponential time shift-rule, since, on many inputs,  it is required to construct almost the whole sequence to find the desired symbol the shift-rule should output. On the other hand, a shift-rule for a DB sequence provides a construction in $O(nk^n)$ time, by finding the next symbols one by one, which is not an optimal approach. 

We see that none of these two methods, a general construction and a shift-rule, dominates the other, and we propose here a third way, which generalizes both methods. A \emph{Generalized-Shift-Rule} ($\GSR$ for abbreviation) for an $(n,k)$-DB sequence is an algorithm that receives two arguments: a word, $w$, of length $n$, and a positive integer $c$. The $\GSR$ outputs the $c$ symbols that follow $w$ at the sequence. Since the algorithm must read its input and must write $c$ symbols, $\Omega(n+c)$ is a trivial time lower bound for a $\GSR$. An optimal $\GSR$ provides an optimal shift-rule when used with $c=1$. In addition, an optimal $\GSR$ provides an optimal construction by invoking it with $c=k^n-n$, or by invoking it $\frac{k^n}{n}$ times with $c=n$ (for example).

Although a $\gsr$ is defined here for the first time, researchers have noted the advantages behind this notion, and mentioned that their shift-rule possesses the properties we seek for in this paper. In~\cite{sawada17} Sawada et al. described a shift-rule with $O(1)$-amortized time per bit. As this seems to contradict the trivial time lower bound mentioned earlier, this statement requires a clarification. The shift-rule proposed in~\cite{sawada17} has the interesting property that after using it once, it can be invoked $c$ more times and, by carefully retaining data from one invocation to another, it can produce the next $c$ symbols in $O(c)$-amortized time. Hence, in fact, Sawada et al. noted and mentioned that their shift-rule also forms a time optimal $\GSR$. A similar remark can be found in~\cite{simple}.

In this paper, we present an optimal $\GSR$ for the well-known prefer-max and prefer-min DB sequences.
Our $\GSR$ construction takes advantage of the seminal block-construction of~\cite{FM} for those sequences, in the following manner.
The sequences are constructed in~\cite{FM} as a concatenation of certain words: $L_1\cdots L_i\cdots$. In Section~\ref{Sec:GSR-FTG} we note that a $\GSR$ can be constructed by solving a similar problem, named \textit{filling-the-gap}. The problem is to find a word $x$ that completes $w$ into a suffix of $L_1\dots L_i$ for some $L_i$. If this suffix is of insufficient length, we use another algorithm, presented in Section~\ref{Sec:Lnext}, which finds the words that follow $L_i$ in that block-construction. In Section~\ref{Sec:FTG} we present a filling-the-gap algorithm which provides, as explained, a $\GSR$ for the prefer-min and prefer-max sequences. In addition, in Section~\ref{Sec:Preliminaries} we define notations used throughout the paper, and conclusions are given in Section~\ref{Sec:Discussion}.

\section{Preliminaries}
\label{Sec:Preliminaries}

For an integer $k$, consider the alphabet $\{0,\dots,k-1\}$, ordered naturally by $0<\cdots < k-1$. Hence, the set of all words over $\{0,\dots,k-1\}$, denoted $\{0,\dots,k-1\}^*$, is totally ordered by the lexicographic order, which we simply denote by `$<$'. As usual, the empty word is denoted by $\varepsilon$. We say that a word, $w$, is an \emph{$m$-word} if $|w|=m$. Furthermore, the \emph{$m$-prefix} of $w$ is the prefix of $w$ of length $m$, and the \emph{$m$-suffix} of $w$ is the suffix of $w$ of length $m$. These notions are defined, of course, only when $|w|\geq m$. 

 A word, $v$, is a \emph{rotation} of a word, $w$, if $w=xy$ and $v=yx$. In addition, $v$ is a \emph{non-trivial rotation} of $w$, if $x\neq w$ and $y\neq w$. Note that a word $w$ can be equal to some of its non-trivial rotations. This happens when $w=x^t$ for some non-empty word, $x$, and an integer $t>1$. In this case, $w$ is said to be \emph{periodic}. A \emph{Lyndon word}~\cite{lyndonwords} is a non-empty word that is strictly smaller than all its non-trivial rotations. Hence, in particular, a Lyndon word is aperiodic.

The prefer-max $(n,k)$-DB sequence is the cyclic sequence  constructed by the greedy algorithm that starts with $0^n$, and repeatedly adds the largest possible symbol in $\{0,\dots,k-1\}$ so that no $n$-word appears twice as a subword of this sequence, until the sequence length is $k^n$, and then rotates the obtained sequence to the left $n$ times. As an example, for $n=2$ and $k=3$, this greedy process produces the sequence: $002212011$ and the prefer-max $(2,3)$-DB sequence is: $221201100$. Analogously, the prefer-min $(n,k)$-DB sequence is produced by the greedy algorithm which starts with $(k-1)^n$, and repeatedly concatenates the smallest possible symbol so that no repetition occurs, and afterwards rotates the resulting sequence to the left $n$ times. 

We note that the prefer-min $(n,k)$-DB sequence and the prefer-max $(n,k)$-DB sequence can be derived one from the other by replacing each symbol, $m$, with $k-1-m$.  Therefore, a $\GSR$ for one of these sequences can be easily transformed into a $\GSR$ for the other one as well. We present here a $\GSR$ algorithm for the prefer-min $(n,k)$-DB sequence. 

From this point on, we refer to $n$ and $k$ as fixed, unknown, parameters, larger than $1$ (to avoid trivialities). We measure time complexity of all algorithms given here in terms of the parameter $n$, assuming that arithmetic operations can be computed in constant time, regardless of how large the numbers they are applied on. 

Let  $\class L_1,\class L_2,\class L_3,\dots$ be the (finite) sequence of all Lyndon words over $\{0,\dots,k-1\}$ of length at most $n$, sorted lexicographically. Let $N(n,k)$ be the number of all Lyndon words over $\{0,\dots,k-1\}$ whose length divides $n$, and let $L_1,L_2,\dots, L_{N(n,k)}$ be an enumeration of them, sorted lexicographically. For a Lyndon word, $L_i$, let $r_i=\frac{n}{|L_i|}$. Since $|L_i|$ divides $n$, $r_i$ is a positive integer. Note that every $L_i$ is equal to some $\class L_j$ such that $j\geq i$. The main result of~\cite{FM} (with a straightforward adaptation) is:

\begin{theorem}
\label{Theorem:L0L1PrefMin}
The prefer-min $(n,k)$-DB sequence is: $L_1 \cdots L_{N(n,k)}$.
\end{theorem}

 As an example, for $n=k=3$ we concatenate in an increasing order all Lyndon words of length one or three. We get the following sequence, decomposed into Lyndon words:
 $$\text{prefer-min $(3,3)$-DB}=0|001|002|011|012|022|1|112|122|2$$

As said, our strategy in constructing a $\GSR$ for the prefer-min sequence is to fill the gap between the input, $w$, to a word $L_i$ in the sequence, and then to concatenate Lyndon words until we find the required $c$ symbols that follow $w$. To this end, we refer to the sequence $L_1,\dots,L_{N(n,k)}$ as cyclic, meaning that,
for $t\in \mathbb{Z}$ and $i\in \{1,\dots,N(n,k)\}$, we set $L_{t\cdot N(n,k)+i}=L_i$.

\section{An Efficient {\sf {Lnext}} Algorithm  }
\label{Sec:Lnext}

As a first step in constructing a $\GSR$ algorithm, we analyze a relatively simple case. By Theorem~\ref{Theorem:L0L1PrefMin}, for every $i\leq N(n,k)$ the sequence $L_1\cdots L_i$ is a prefix of the prefer-min  sequence. We consider the case where we are given an $n$-word, $w$, that happens to be a suffix of $L_1\cdots L_i$. 
To find the next $c$ symbols, we can compute the next words in the block construction, $L_{i+1},\cdots L_{i+j}$, so that the sequence $L_{i+1}\cdots L_{i+j}$ is of length at least $c$.

For dealing with this restricted case, we design an algorithm that computes efficiently the function: $\Lnext(L_i)=L_{i+1}$. Moreover, for technical reasons that will arise later, we also want to apply the algorithm over Lyndon words whose length does not necessarily divide $n$. Therefore, for a Lyndon word,
$\class L_j\neq L_{N(n,k)}$, we define $\Lnext(\class L_j)$ to be the lexicographically smallest $L_i$ such that
$\class L_j<L_i$. For $\class L_j = L_{N(n,k)}$, we define   $\lnext(L_{N(n,k)})=L_1$ (that is, $\Lnext(k-1)=0$).
In this section, we present an $\Lnext$ algorithm with $O(n)$ time complexity.

\begin{proposition}
$\forall \mathcal{L}_i$ Algorithm~\ref{Algo_next_n_Lyndon} computes $\Lnext(\mathcal{L}_i)$ in $O(n)$ time.
\label{prop_next_n_Lyndon}
\end{proposition}

In~\cite{duval1988generation}, Duval describes an algorithm to build the next Lyndon word from a given one and proves:
\begin{theorem}
	\label{theoram_Duval_Li_Li}
	On every input $\class L_i \neq k-1$, Algorithm~\ref{Algo_Duval} returns $\class L_{i+1}$ in $O(n)$ time. 
\end{theorem}

\begin{algorithm}[ht]
\caption{\Duval}
\textit{Input:} A Lyndon word, $\class L_i\neq k-1$\\
\textit{Output:} $\class L_{i+1}$

\hrulefill
\begin{algorithmic}[1]

\State $x\gets (\class L_i)^tu$, where $u$ is a proper prefix prefix of $\class L_i$ and $|(\class L_i)^tu|=n$
\State remove largest suffix of $x$ of the form $(k-1)^l$
\State increase the last symbol of $x$ by 1
\State return $x$

\end{algorithmic}
\label{Algo_Duval}
\end{algorithm}

Note that the length of the output of the algorithm may not divide $n$. However, we use it to construct a naive algorithm which achieves that goal, with a time complexity of $O(n^2)$. We first describe this naive version, which merely invokes Algorithm~\ref{Algo_Duval} several times, and then improve it to run in linear time.

\begin{algorithm}[h]
\caption{\sf{Naive-Lnext}}
\textit{Input:} a Lyndon word, $\mathcal{L}_i$\\
\textit{Output:} $\lnext(\class L_i)$

\hrulefill

\begin{algorithmic}[1]

\If{$\class L_i=k-1$}
\State return $0$ \EndIf

\State   $x\gets \Duval(\mathcal{L}_i)$ 

\If {$|x| \mid n$} 
\State return $x$
\EndIf 

\If {$|x|< \frac{n}{2}$} 
\State $x\gets \Duval(x)$
\EndIf \hfill// now $|x|>\frac{n}{2}$ 

\State $\mathit{length} \gets |x|$

\While{$\mathit{length} \nmid n$}  
\State $x\gets \Duval(x)$, $\mathit{length}\gets |x|$\EndWhile

\State return $x$. 
\end{algorithmic}
\label{Algo_Naive-Next-Lyndon}
\end{algorithm}

Note that $\forall \class L_i$, Algorithm~\ref{Algo_Naive-Next-Lyndon} outputs $\lnext(\class L_i)$. At each iteration of the loop in lines 12-14, the algorithm invokes Duval's algorithm, until it finds a Lyndon word whose length divides $n$. This establishes a worst case runtime of $\Theta(n^2)$. 

The reader may note that the  {\bfseries if} instructions in lines 5-7 and lines 8-10 can be omitted. However, we aim to construct a linear time $\Lnext$ algorithm, and we do that by modifying the {\bfseries while} loop. Then, it will be important that the loop acts on words whose length is larger than $\frac{n}{2}$. Thus, lines 5-7 and 8-10 are added to simplify the comparison between this naive $\lnext$ version and our linear time $\lnext$ version.

To improve the runtime of this algorithm, we identify cases in which the outcome of several loop iterations can be computed directly. These are the cases in which calling Duval's algorithm again and again results in concatenating the same sequence several times. For illustration, assume that at some point the algorithm reaches line 12 when $x$ stores a word $\class L_j=00(k-1)^l$, such that $|\class L_j| = 2+l \geq \frac{n}{2}$. The Lyndon words that follow $\class L_j$ are:
$$00(k-1)^l 01, 00(k-1)^l 0101,\dots,00(k-1)^l(01)^{\lfloor \frac{n-l-2}{2} \rfloor}.$$ Instead of applying $\Duval$ $t=\lfloor \frac{n-l-2}{2} \rfloor$ times, we can save time by computing $t$, and go to $00(k-1)^l(01)^t$ without traversing all words in that list. This allows us to compute $\lnext(\class L_i)$ in linear time, as we do in Algorithm~\ref{Algo_next_n_Lyndon}.

\begin{algorithm}[ht]
\caption{\Lnext}

\textit{Input:} a Lyndon word,  $\mathcal{L}_i$\\
\textit{Output:} $\lnext(\class L_i)$

\hrulefill
\begin{algorithmic}[1]
 \If{$\class L_i=k-1$}
\State return $0$ \EndIf

\State  $x\gets \Duval(\class L_i)$
\If {$|x| \mid n$} 
\State return $x$
\EndIf 

\If {$|x|< \frac{n}{2}$}
\State $x\gets \Duval(x)$
\EndIf \hfill // now $|x| > \frac{n}{2}$ 

\State $\mathit{length} \gets |x|$

\While {$\mathit{length}\nmid n$}   
  \State $m\gets n-\mathit{length}$
 
 \State $u\gets$ prefix of $x$ of length $m$
 
 \State $u'\gets$ the word obtained from $u$ by removing its suffix that includes only occurrences of $k-1$, and increasing the last symbol by $1$ 
 
 \State $t\gets\lfloor \frac{m}{|u'|} \rfloor$
 
 \State $x\gets xu'^{t}$, $\mathit{length}\gets \mathit{length}+|u'^t|$
 \EndWhile
 
 \State return $x$
\end{algorithmic}
\label{Algo_next_n_Lyndon}
\end{algorithm}

In order to prove Algorithm~\ref{Algo_next_n_Lyndon} correctness, we show that both algorithms, Algorithm~\ref{Algo_Naive-Next-Lyndon} and Algorithm~\ref{Algo_next_n_Lyndon}, have the same output for every legal input. We start with two observations, derived from Duval's algorithm.

\begin{corollary}
\label{corollary_L_i-L_i+1}
$\forall \mathcal{L}_i\neq k-1$, if $|\class L_i|<n$, then $|\class L_i|<|\class L_{i+1}|$.
\end{corollary}

 \begin{corollary}
\label{corollary_Lyndon_length_above_half_n}
$\forall \mathcal{L}_i\neq k-1$, if $|\mathcal L_i| < n$, then $|\mathcal L_{i+1}|>\frac{n}{2}$.
\end{corollary}

In sketch, these observations are proved as follows: If $|\mathcal{L}_i|<n$, then $\mathcal{L}_{i+1}=(\mathcal{L}_i)^mz$  for some $m>0$ and $z\neq \varepsilon$. Clearly, Corollary~\ref{corollary_L_i-L_i+1} follows. Also, note that  that $|(\mathcal{L}_i)^m|\geq \frac{n}{2}$, which proofs Corollary~\ref{corollary_Lyndon_length_above_half_n}.

From these two observations we can deduce the following conclusion, discussing the similarity between the two algorithms when entering the {\bfseries while} loop:

\begin{corollary}
\label{Corollary_invariant}
The following invariant holds for both Algorithm~\ref{Algo_Naive-Next-Lyndon} and Algorithm~\ref{Algo_next_n_Lyndon} whenever the {\emph{\bfseries while}} loop starts:  $|x| > \frac{n}{2}$.

\end{corollary}

The next lemma shows that every execution of the {\bfseries while} loop in Algorithm~\ref{Algo_next_n_Lyndon} corresponds to several executions of the {\bfseries while} loop in Algorithm~\ref{Algo_Naive-Next-Lyndon}.

\begin{lemma}
\label{Lemma:Duval_X_u-j_is_X_u-j+1}
Let $x$ be a Lyndon word such that $|x|>\frac{n}{2}$. Let $m, u, u', t$ be as in lines 13-16 of Algorithm~\ref{Algo_next_n_Lyndon}. Then, for $j\leq t$:
\begin{enumerate}
\item $xu'^j$ is a Lyndon word.
\item If $j<t$, then $\Duval(xu'^j)=xu'^{j+1}$.
\end{enumerate}
\end{lemma}
\begin{proof}The proof is by induction on $j$. If $j=0$, item 1 holds as $x$ is a Lyndon word, and if $j>0$ item 1 
holds by applying the induction hypothesis on $j-1$, and then using item 2 of the lemma.

It remains to prove that item 2 holds thus suppose $j<t$. Write $u'=v(\sigma+1)$ for a word $v$ and a symbol $\sigma+1 < k$. Hence, the $m$-prefix $u$ of $x$ is $u=v\sigma (k-1)^l$ for some $l\geq 0$. 
Namely, $x=v\sigma(k-1)^lv'$ for some word, $v'$, such that $m=n-|x|=|u|=|v\sigma(k-1)^l|$ (note that since $|x|>\frac{n}{2}$, the $m$-prefix of $x$ is defined). 
Thus,  $xu'^j=v\sigma(k-1)^lv'(v(\sigma+1))^j$. 
Let $\hat{m}=n-|xu'^j|\leq m= |u| = |v\sigma(k-1)^l|$. Note that since $j<t$ it follows that $\hat{m}\geq |u'|=|v\sigma|$ (this holds because $t$ is the largest integer such that $|xu'^{t}|\leq n$). 

To summarize, $\hat{m}\leq m=|v\sigma(k-1)^l|$ and $\hat{m}\geq |v\sigma|$. Hence, the $\hat{m}$-prefix of $xu'^{j}=v\sigma (k-1)^l v'u'^{j}$ is $v\sigma(k-1)^{l'}$ for some $l'\leq l$.
Therefore, $\Duval(xu'^j)$ is obtained by concatenating $v\sigma(k-1)^{l'}$ to $xu'^j$, removing the suffix $(k-1)^{l'}$ and increasing $\sigma$ by one. Namely, $\Duval(xu'^j)=xu'^jv(\sigma+1)=xu'^{j+1}$, as required.
\end{proof}

Lemma~\ref{Lemma:Duval_X_u-j_is_X_u-j+1} states that each execution of the {\bfseries while} loop of Algorithm~\ref{Algo_next_n_Lyndon} corresponds to $t$ executions of the {\bfseries while} loop of Algorithm~\ref{Algo_Naive-Next-Lyndon}. Therefore, we conclude:
\begin{corollary}
\label{Corollary_next-Lyndon_is_Naive-Next-Lyndon}
$\forall \mathcal{L}_i$, the output of Algorithm~\ref{Algo_Naive-Next-Lyndon} over the input $\class L_i$, is equal to the output of Algorithm~\ref{Algo_next_n_Lyndon} over the input $\class L_i$.
\end{corollary}

It is left to prove that our runtime is linear. For this purpose, consider an execution of Algorithm~\ref{Algo_next_n_Lyndon} on input $\class L_i$, and assume that the execution reaches line 12 (otherwise, the algorithm terminates in line 2 after $O(1)$ steps, or in line 6 after $O(n)$ steps). We need to show that the loop terminates after $O(n)$ steps. This fact follows from the following observation. The loop terminates when $|x|=n$, and after each loop iteration the value $n-|x|$ decreases by at least half. 

%To prove this fact in a precise manner, let us denote by $m_i,u_i,u'_i,t_i$ and $x_i$ the values assigned to variables $m,u,u',t$ and $x$, respectively, at the $i$-th iteration of the {\bfseries while} loop. In addition, let $x_0$ be the value of variable $x$ before entering the {\bfseries while} loop for the first time, and let $r$ denote the number of loop iterations of the execution.

To prove this fact in a precise manner, we introduce a few notations.
\begin{enumerate}
    \item $m_i,u_i,u'_i,t_i,$ and $x_i$ denote the values assigned to variables $m,u,u',t,$ and $x$, respectively, at the $i$-th iteration of the {\bfseries while} loop.
    
    \item $x_0$ is the value that variable $x$ stores before entering the {\bfseries while} loop for the first time.
    
    \item $r$ is the number of iteration of the {\bfseries while} loop .
\end{enumerate}

\begin{lemma}
\label{Lemma_m_i_log}

For $0<i\leq r$, we have $m_i=n-|x_{i-1}|\leq \frac{n}{2^i}$.
\end{lemma}

\begin{proof}
By induction on $i$. 
The base case $i=1$ is trivial, as $|x_0|>\frac{n}{2}$. For the induction step, take $m_{i+1}=n-|x_i|$ and note that $x_i=x_{i-1}(u_i')^{t_i}$, where $|u_i'|\leq |u_i|=m_i$ and $t_i$ is the largest integer so that $|(u_i')^{t_i}|\leq m_i$. Therefore, $|(u_i')^{t_i}|\geq \frac{m_i}{2}$ and hence, by the induction hypothesis, we get:
$|m_{i+1}|=n-|x_i|=n-|x_{i-1}(u_i')^{t_i}|=n-|x_{i-1}|-|(u_i')^{t_i}|=m_i-|(u_i')^{t_i}|\leq m_i-\frac{m_i}{2}=\frac{m_i}{2}\leq \frac{n}{2^{i+1}}$.
\end{proof}

Relying on this lemma, we can now analyze the runtime of our algorithm and prove Proposition~\ref{prop_next_n_Lyndon}.

\begin{proof}[Proof of Proposition~\ref{prop_next_n_Lyndon}]
By Corollary~\ref{Corollary_next-Lyndon_is_Naive-Next-Lyndon}, the algorithm computes $\lnext(\class L_i)$ correctly. We shall prove that the algorithm runs in $O(n)$ time. If the algorithm returns in line 2 or in line 6, the execution terminates after $O(n)$ time, and we are done.  Otherwise,
 by Lemma~\ref{Lemma_m_i_log}, we have:
 \begin{enumerate}
     \item $r\leq \log(n)$.
     \item For $0<i\leq r$ it holds that $m_i\leq \frac{n}{2^i}$.
 \end{enumerate}

The latter explicitly appears in Lemma~\ref{Lemma_m_i_log}, and the former is argued as follows: Assume towards  a contradiction that $r>\log(n)$, and consider the $(\log(n)+1)$ iteration of the {\bfseries while} loop in line 12. By Lemma~\ref{Lemma_m_i_log}, $m_{\log (n) +1}\leq \frac{n}{2^{\log(n) +1}}$ and thus $m_{\log( n) +1}=0$. Hence, at the $(\log (n)+1)$ iteration, variable $\mathit{length}$ stores $n$ (see line 13). Thus, the  {\bfseries while} loop pre-condition is not satisfies  (see line 12), which contradicts the assumption that the $(\log(n)+1)$ iteration was performed.

Now, in each iteration of the {\bfseries while} loop, finding $u$ and $u'$ are the most time-consuming steps, each costs $O(m_i)\leq O(\frac{n}{2^i})$. Therefore, the global runtime is $O(n+m_1+m_2+\dots+m_{r}) \leq O(n+\frac{n}{2}+ \frac{n}{4}+ \dots +\frac{n}{2^{r}}) \leq O(n+\frac{n}{2}+ \frac{n}{4}+ \dots +\frac{n}{2^{log(n)}})= O(n)$.
\end{proof}

\section{A $\GSR$ Algorithm Based on a Reduction to $\ftg$}
\label{Sec:GSR-FTG}

The fact that $\Lnext$  can be computed efficiently is useful for designing an efficient $\GSR$ algorithm. Given an $n$-word, $w$, assume that $wx$ is a suffix of $L_1\cdots L_i$. In this case, several invocations of our $\Lnext$ algorithm produce the $c$-word that follows $w$ at the prefer-min sequence. For taking this approach, first, it is required to find a Lyndon word, $L_i$, and a word, $x$, such that $wx$ is a suffix of $L_1\cdots L_i$. This implies that a $\gsr$ algorithm for the prefer-min sequence can be derived from a solution to another problem we propose in this section: {\it Filling-The-Gap} ($\FTG$ for abbreviation).

\begin{definition}
\label{def:ftg(w)}
For an $n$-word $w$ we write $\ftg(w)=(L_i,x)$, if the following hold:
\begin{enumerate}
    \item $wx$ is a suffix of $L_1\dots L_i$.
    \item If $wy$ is a suffix of $L_1\dots L_j$ for some $j$, then $|x|\leq |y|$. 
\end{enumerate}

\end{definition}

We leave for the reader to verify that $\ftg(w)$ is well-defined, meaning that for every $n$-word, $w,$ only a single pair, $(L_i,x)$, satisfies the conditions of Definition~\ref{def:ftg(w)}. We remark that it is possible that $\ftg(w)=(L_i,x)$ where $i>N(n,k)$. This occurs in the case where $w$ is a concatenation of a suffix of the prefer-min sequence with a prefix of it. For example, if $w=(k-1)0^{n-1}$, then $\ftg(w)=(L_{N(n,k)+2},01)$ since $w01$ is a suffix of $L_1\cdots L_{N(n,k)}L_{N(n,k)+1}L_{N(n,k)+2}=L_1\cdots L_{N(n,k)}L_1L_2$.

Note that $\FTG (w)$ can be trivially computed by concatenating Lyndon words and searching for $w$. However, this naive solution is highly inefficient as $w$ may appear anywhere in the prefer-min sequence. Hence, for constructing an efficient $\GSR$ in the way described above, we need an efficient $\FTG$-algorithm.

There is also another issue concerning the suggested approach, which requires attention. If $\ftg(w)=(L_i,x)$, for computing the $c$-word that comes after $w$ we need to invoke Algorithm~\ref{Algo_next_n_Lyndon} several times. It is required to explain why the number of Lyndon words we concatenate is proportional to the suffix we seek for. More precisely, we need to show that the total number of invocations of Algorithm~\ref{Algo_next_n_Lyndon} consumes $O(n+c)$ time. This is settled by the next lemma, which  claims that there are no two consecutive words, $L_i,L_{i+1}$,  both of length smaller than $n$:

\begin{lemma}
\label{Lemma:L_i-L_i+1_sizes}
For $1\leq i<N(n,k)$, if $|L_i|<n$, then $|L_{i+1}|=n$. 
\end{lemma}

\begin{proof}
Write $L_i=\mathcal{L}_j$ and $L_{i+1}=\mathcal{L}_{j+t}$. That is, 
 $t$ is the smallest (positive) integer such that $|\mathcal{L}_{j+t}|$ divides $n$.
By Corollary~\ref{corollary_Lyndon_length_above_half_n}, $|\mathcal{L}_{j+1}|>\frac{n}{2}$.
By Corollary~\ref{corollary_L_i-L_i+1}, $|\mathcal{L}_j|<\cdots<|\mathcal{L}_{j+t}|$. Therefore,  $|L_{i+1}|=|\mathcal{L}_{j+t}|=n$.
\end{proof}

We can now present, in Algorithm~\ref{algo_generalized-shift-rule}, a $\GSR$ algorithm based on a reduction to the $\ftg$ problem.

\begin{algorithm}[ht]
\caption{{\sf generalized\_shift\_rule }}
\textit{Input:} ($w,c$), $|w|=n$\\
\textit{Output:}  a word $w'$ of length $c$ that appears after $w$ at prefer-min

\hrulefill
\begin{algorithmic}[1]

\State $(L,x)\gets \ftg(w)$

\State $\mathit{length}\gets |x|$

\While{$\mathit{length}<c$} 
\State $L\gets \Lnext(L)$

\State $x\gets xL$

\State $\mathit{length}\gets \mathit{length}+|L|$\EndWhile
\State return the $c$-prefix of $x$
\end{algorithmic}
\label{algo_generalized-shift-rule}
\end{algorithm}

Consider the {\bfseries while} loop in algorithm~\ref{algo_generalized-shift-rule} and use Lemma~\ref{Lemma:L_i-L_i+1_sizes} to conclude that after $O(1)$ loop iterations, which consume $O(n)$ time, $|x|$ increases by at least $n$ symbols. It follows that the loop halts in $O(n+c)$ steps and hence, we get the following:

\begin{proposition}
\label{prop:main-prop-of-sec4}
If $FTG(w)$ can be computed in
 $O(n)$ time, then Algorithm~\ref{algo_generalized-shift-rule} forms a $\GSR$ for the prefer-min $(n,k)$-DB sequence with $O(n+c)$ time complexity.
 \end{proposition}

\section{An $\FTG$ Algorithm}
\label{Sec:FTG}

In this section we construct an efficient $\ftg$-algorithm. This is done in two steps. First, we define the notion of a \emph{cover} of an $n$-word, $w$, and show how a cover for $w$ can be transformed into $\ftg(w)$ efficiently. Then, we show how to find a cover for an $n$-word, $w$, in linear time.

\subsection{Finding $\FTG(w)$ by Means of a Cover}
\label{Subsec:cover2ftg}

The $\FTG$ problem, applied on an $n$-word, $w$, is to extend $w$ into a suffix of $L_1\cdots L_i$. 
For solving this problem, we introduce a similar notion.

\begin{definition}
\label{def:cover(w)}
For an $n$-word, $w\neq (k-1)^p0^{n-p}$, $\cover(w)=(L_i,x)$ if the following hold: 
\begin{enumerate}
    \item $wx$ is a suffix of $L_1\dots L_{i-1}L_i^{r_i}$.
    \item If $wy$ is a suffix of $L_1\dots L_{j-1}L_j^{r_j}$ for some $j$, then $|x|\leq |y|$. 
\end{enumerate}
In addition, we say that $w$ is covered by $L_i$, if $\cover(w)=(L_i,x)$ for some word $x$. 
\end{definition}

 Also here, we leave for the reader to verify that $\cover(w)$ is well-defined. We focus on $n$-words different from $(k-1)^p0^{n-p}$ from technical reasons, as it allows us to provide a simpler presentation of our results. Otherwise, many parts in our analysis should be rephrased, and some proofs should be rewritten, to include more details. However, it is simple to show that the $\ftg$ algorithm we provide at the end of this section, works for every $n$-word.

The two notions, $\cover(w)$ and $\ftg (w)$ are closely related.
The difference between these notions can be bridged by observing that if $(L_i,x)$ is a cover for $w$, then $w$ is a subword of $L_1\cdots L_{i+1}$. To clear this issue, we deal with the  relationships between two consecutive Lyndon words in the following Lemma.

\begin{lemma}
\label{Lemma:L_i_to_L_i+1}
For all $i$ such that $1\leq i<N(n,k)$, if $|L_i|<n$, then $L_{i+1}=L_i^{r_i-1}z$, for some $|L_i|$-word, $z$.
\end{lemma}

\begin{proof}
Let $L_i=\class L_j$ and $L_{i+1}=\class L_{j+m}$. By an examination of Algorithm~\ref{Algo_Duval}, and since $L_i$ does not start with $k-1$ (the only such word is $L_{N(n,k)}$, and $1\leq i<N(n,k)$), we conclude that $\class L_{j+1}=L_i^{r_i-1}u_1$, where $u_1$ is constructed by removing the suffix of $L_i$ that includes only occurrences of $k-1$, and increasing the last symbol of the obtained word by one. A simple inductive argument shows that  for $t\in\{1, \dots, m\}$ we have $\class L_{j+t}=L_i^{r_i-1}u_1\dots u_t$ for some non-empty words, $u_1,\dots,u_t$. By Lemma~\ref{Lemma:L_i-L_i+1_sizes}, $|L_{i+1}|=n$ thus $L_{i+1}=\class L_{j+m}=L_i^{r_i-1}u_1\dots u_m$, and $z=u_1\dots u_m$ is an $|L_i|$-word.  
\end{proof}

It follows that for $1\leq i<N(n,k)$ we have that $L_1\cdots L_{i-1}L_i^{r_i}$ is a prefix of prefer-min. This trivially holds when $r_i=1$, and if $r_i>1$, the  previous lemma  ensures that there exists some $z$ such that $L_1\cdots L_iL_i^{r_i-1}z = L_1\cdots L_{i+1}$. Now we turn to deal with the relationships between $\cover(w)$ and $\FTG(w)$.

\begin{lemma}
\label{Lemma:i-to-i+1}
Assume that $\cover(w)=(L_i,x)$ for an $n$-word, $w\neq (k-1)^p0^{n-p}$. 
\begin{enumerate}

    \item  
    When $|x|\geq |L_i^{r_i-1}|$, we have $\ftg(w)=(L_i,y)$ where $y$ is the $(|x|-|L_i^{r_i-1}|)$-prefix of $x$. 

    \item Otherwise, $\ftg(w)=(L_{i+1},xz)$, where $z$ is the $|L_i|$-suffix of $L_{i+1}$.
\end{enumerate}
\end{lemma}

\begin{proof}
We start by proving the first item. As $\cover(w)=(L_i,x)$, $wx$ is a suffix of $L_1\cdots L_{i-1}L_i^{r_i}$. Since $|x|\geq |L_i^{r_i-1}|$, $wy$ is a suffix of $L_1\cdots L_i$. Moreover, the minimality of $x$ guarantees that $w$ is not a subword of $L_1\cdots L_{i-1}$ thus $\ftg(w)=(L_i,y)$ as required.

We turn to prove the second item, in which $wx$ is a suffix of $L_1\cdots L_{i-1}L_i^{r_i}$ and $|x|< |L_i^{r_i-1} |$. Thus, it follows that $|L_i|<n$, since otherwise $r_i=1$ and we get the false equation: $|x|<|\varepsilon|$.  Moreover, as  $w\neq (k-1)^p0^{n-p}$, it follows that $1 < i< N(n,k)$. Hence, Lemma~\ref{Lemma:L_i_to_L_i+1} can be invoked and we get that  $L_{i+1}=L_i^{r_i-1}z$ where $z$ is the $|L_i|$-suffix of $L_{i+1}$. As a result, $wxz$ is a suffix of $L_1\cdots L_{i-1}L_i L_i^{r_i-1}z=L_1\cdots L_i L_{i+1}$. Furthermore, since $wx$ is a suffix of $L_1\cdots L_{i-1}L_i^{r_i}$ and $|x|<|L_i^{r_i-1}|$, we conclude that $w$ is not a suffix of $L_1\cdots L_i$. Therefore, $\ftg(w)=(L_{i+1},xz)$ as required.
\end{proof}

 Using the above, Algorithm~\ref{algo:cover2ftg} transforms $\cover(w)$ into $\ftg(w)$ in linear time.

\begin{algorithm}[ht]
\caption{\ctf}
\label{cover-to-ftg}
\textit{Input:} a pair, $(L,x)=\cover(w)$ \\
\textit{Output:} $\ftg(w)$ 

\hrulefill

\begin{algorithmic}[l]

\State $r\gets$ $n/|L|$
\If  {$|x|\geq |L^{r-1}|$}
\State $y\gets (|x|-|L^{r-1}|)$-prefix of $x$
\State return $(L, y)$ \EndIf
\State $z\gets|L|$-suffix of  $\Lnext(L)$
\State return $(\Lnext(L),xz)$

\end{algorithmic}
\label{algo:cover2ftg}
\end{algorithm}

We conclude this subsection with the next corollary. Its first item follows by Lemma~\ref{Lemma:i-to-i+1}, and its second item follows from the code; for the input $(0,0^p)$, line 1 assigns $n$ to variable $r$, and the condition in line 2 does not hold. Then, $z$ is assigned with $1$, and the algorithm returns $(0^{n-1}1,0^p1)$, which is $\ftg(w)$.

\begin{corollary}
\label{cor:main-cor-of-5.1}
Let $w\neq(k-1)^n$ be an $n$-word.
\begin{enumerate}
    \item If $w\neq (k-1)^p0^{n-p}$ and $\cover(w)=(L,x)$, then Algorithm~\ref{algo:cover2ftg} returns $\ftg(w)$ on the input $(L,x)$.
    
    \item If $w=(k-1)^p0^{n-p}$, then Algorithm~\ref{algo:cover2ftg} returns $\ftg(w)$ on the input $(0,0^p)$.

\end{enumerate}

\label{Cor:cover2ftg}
\end{corollary}

\subsection{Computing $\cover(w)$}
\label{Subsec:finding-a-Cover}

In this section we show how to compute $\cover(w)$, efficiently. Assume that an $n$-word, $w\neq (k-1)^p0^{n-p}$, is covered by $L_{i+1}$. Thus, $w$ is a subword of $L_iL_{i+1}^{r_{i+1}}$. To compute $\cover(w)$ (in Algorithm~\ref{Algo:filling_the_gap}), in some cases, we compute $L_i$ and use it to find $\cover(w)$, and in other cases we compute directly the suffix of $L_iL_{i+1}^{r_{i+1}}$ that follows $w$. The way this goal is achieved relies on the analysis we provide here, which we divide into two parts. First, we show how to construct $L_{i+1}^{r_{i+1}}$ from $L_i^{r_i}$, by concatenating certain words to $L_i^{r_i-1}$. Then, we present a structural characterization of $w$ which will serve us to compute  $\cover(w)$.

\subsubsection{Modifying $L_{i}^{r_{i}}$ into $L_{i+1}^{r_{i+1}}$ }

Assume that an $n$-word, $w\neq (k-1)^p0^{n-p}$, is covered by $L_{i+1}$. Hence, $w$ is a subword of $L_1\cdots L_i L_{i+1}^{r_{i+1}}$, but not of $L_1\cdots L_{i-1}L_i^{r_i}$.
Clearly, Lemma~\ref{Lemma:L_i_to_L_i+1} implies that  $L_1 \cdots L_{i-1}L_i^{r_i}$ is a prefix of $L_1\cdots L_i L_{i+1}^{r_{i+1}}$, but what is the difference between these two sequences? The first goal of our analysis is to show how to construct $L_{i+1}^{r_{i+1}}$ from $L_{i}^{r_{i}}$, by concatenating a suffix to $L_{i}^{r_{i}}$.

\begin{definition}
\label{Def_vi}
Let $L_i\neq k-1$. We define a sequence of words: $v_{i,1},\dots, v_{i,m_i}$ and a sequence of indices: $k_1,\dots,k_{m_i+1}$ by induction, where $k_i$ indicates the amount of characters left to calculate in step $i$:
Write $k_1=|L_i|$ and assume that $v_{i,1}\dots v_{i,j-1}$ were defined, together with  $k_1,\dots,k_j$.
\begin{itemize}
    \item If $k_j=0$, then $j-1=m_i$ and we are done.
    \item Otherwise, $v_{i,j}$ is obtained as follows: take the prefix of $L_i$ of size $k_j$, remove its suffix that includes only occurrences of $k-1$, and increase the last symbol by one. In addition, let $k_{j+1}=k_j-|v_{i,j}|=|L_i|-|\vv i1|-\cdots -|\vv ij |$.
\end{itemize}
\end{definition}

As an illustration of this definition, we give the following example:
\begin{example}
 Let $n=7$, $k=2$ and $L_i=0010111$. Then,
\begin{itemize}
    \item $k_1=7$ and $\vv  i 1=0011$.
    
    \item $k_2=3$ and $\vv i 2=01$.
    
    \item $k_3=1$ and $\vv i 3 =1$. Moreover, $k_4=0$ thus $m_i=3$ and the process is completed. 
    
\end{itemize}
Also, the reader may check that $L_{i+1}=\vv i 1\vv i 2 \vv i 3= 0011011$, as  Corollary~\ref{cor:main-cor-of-5.2.1} states.  
\end{example}

We show now how the words $\vv i1,\dots,\vv i {m_i}$ form as building blocks for constructing $L_{i+1}^{r_{i+1}}$ from $L_i^{r_i}$. We divide the analysis into three Lemmas, to deal with the different cases.

\begin{lemma}
\label{Lemma_vi_case1}
Take $L_i\neq k-1$, and consider the words $\vv i1,\dots,\vv i {m_i}$, as defined in Definition~\ref{Def_vi}. If $r_i>1$, then: 
\begin{enumerate}
\item $\Duval(L_i)=L_i^{r_i-1}\vv i 1$.
\item For $1\leq j< m_i$, it holds that $\Duval(L_i^{r_i-1}\vv i 1 \cdots \vv i j )=L_i^{r_i-1}\vv i 1 \cdots \vv i j \vv i {j+1}$. 
\end{enumerate}
\end{lemma}

\begin{proof}
The first item follows immediately from the definition of $\vv i 1 $. For proving the second item, note that $k_{j+1}=|L_i|-|\vv i 1 \cdots \vv ij|$, and since $|L_i^{r_i}|=n$, we have $k_{j+1}=n-|L_i^{r_i-1}\vv i1 \cdots \vv i j|$. Let $v\sigma(k-1)^l$ be the $k_{j+1}$-prefix of $L_i$, where $\sigma<k-1$. Hence, $\Duval(L_i^{r_i-1}\vv i1 \cdots \vv i j)=L_i^{r_i-1}\vv i1 \cdots \vv i j v(\sigma+1)$. Also, by Definition~\ref{Def_vi}, $\vv i {j+1}=v(\sigma+1)$, which completes the proof. 
\end{proof}

\begin{lemma}
\label{Lemma_vi_case2}
Take  $L_i\neq k-1$, and consider the words $\vv i1,\dots,\vv i {m_i}$, as defined in Definition~\ref{Def_vi}. If $r_i=1$ and $\vv i 1 \mid n$, then: 
\begin{enumerate}
\item $\Duval(L_i)=\vv i 1=L_{i+1}$.
\item $\vv i 1 = \cdots= \vv i {m_i}$ and $m_i=r_{i+1}$.  
\end{enumerate}
\end{lemma}

\begin{proof}
For the first item, note that by Definition~\ref{Def_vi} and by the fact that $r_i=1$ (namely, $|L_i|=n$), $\Duval(L_i)=\vv i 1$. Moreover, $\vv i 1 = L_{i+1}$, as $\vv i 1 \mid n$. We turn to prove the second item. Write $L_i=v\sigma(k-1)^l$, where $\sigma<k-1$. Hence, since $k_1=n$, it follows that $\vv i 1 =v(\sigma+1)$. We show now by induction that $\vv i 1=\vij$ for every $j\in\{1,\dots,m_i\}$. The induction basis trivially holds, as $\vv i 1=\vv i 1$. Assume, now, that $j>1$ and $\vv i1 =\cdots = \vv i {j-1}$. Since $j\leq m_i$, it follows that $k_j>0$. Specifically, $k_j=|L_i|-|\vv i 1 |-\cdots -|\vv i {j-1}| = n-(j-1)|\vv i 1|>0 $. Since $|\vv i 1| \mid n$, it follows that $k_j=n-(j-1)|\vv i 1|\geq|\vv i 1|$. 
Therefore, the $k_j$-prefix of $L_i$ is $v\sigma(k-1)^{l'}$ for some $l' <  l$. It follows that $\vij=v(\sigma+1)=\vv i 1$, as required. Moreover, $k_{m_i+1}=0=n-|\vv i 1\cdots \vv i {m_i}|$. Thus, $n=|\vv i 1 \cdots \vv i {m_i}|=|(\vv i 1)^{m_i}|=|(L_{i+1})^{m_i}|$ which proves that $m_i=r_{i+1}$.
\end{proof}

\begin{lemma}
\label{Lemma_vi_case3}
Take $L_i\neq k-1$, and consider the words $\vv i1,\dots,\vv i {m_i}$, as defined in Definition~\ref{Def_vi}. If $r_i=1$ and $\vv i 1 \nmid n$, then:
\begin{enumerate}
\item $\Duval(L_i)=\vv i 1$.

\item  $\vv i 1=\cdots = \vv i {j_0}$ and  $\Duval(\vv i 1) = \vv i 1\cdots\vv i {j_0} \vv i {j_0+1}$, where $j_0$ is the maximal integer such that $|(\vv i 1)^{j_0}|<n$ (note that equality cannot hold since $|\vv i 1 |\nmid n$).

\item For $j_0$ as defined above, if $j_0<j<m_i$, then $\Duval(\vv i 1\cdots \vv i j)=\vv i 1  \cdots \vv i {j+1}$.
\end{enumerate}
\end{lemma}

\begin{proof}
As in the two former lemmas, the first item trivially holds  by Definition~\ref{Def_vi}. To prove the second item, let $j_0$ be maximal such that $|(\vv i 1)^{j_0}|<n$. Then, $\Duval(\vv i 1)=(\vv i 1)^{j_0}v'$. By using the same argument as in the proof of item 2 of Lemma~\ref{Lemma_vi_case2}, it can be shown that $\vv i 1= \cdots =\vv i {j_0}$. We prove that $v'=\vv i {j_0+1}$. Since $r_i=1$ it follows that $|L_i|=n$ and thus $k_{j_0+1}=n-|\vv i 1\cdots\vv i {j_0}|=n-|(\vv i 1)^{j_0}|$. By the maximality of ${j_0}$ it follows that $k_{j_0+1}<|\vv i 1|$. Let $x\sigma(k-1)^l$ be the $k_{j_0+1}$-prefix of $\vv i 1$, where $\sigma < k-1$. Hence, removing the suffix, $(k-1)^l$, and increasing the last symbol by one results in $x(\sigma+1)=v'$. Now, since $k_{j_0+1}<|\vv i 1|$ it follows that $x\sigma(k-1)^l$ is also the $k_{j_0+1}$-prefix of $L_i$. Therefore, $\vv i {j_0+1}=x(\sigma+1)=v'$ as required. We leave for the reader to verify that the same argument proves item 3 as well.
\end{proof}

Now we can show how to construct $L_{i+1}^{r_{i+1}}$ from $L_i^{r_i}$.

\begin{lemma}
\label{Lemma:L_i-vvv_L_i+1}
Take $L_i\neq k-1$, and consider the words $\vv i1,\cdots,\vv i {m_i}$, as defined in Definition~\ref{Def_vi}. Then, $L_i^{r_i-1}\vv i 1 \cdots \vv i {m_i}= L_{i+1}^{r_{i+1}}$.
\end{lemma}

\begin{proof}
The proof is divided into three parts, in accordance with Lemmas~\ref{Lemma_vi_case1},~\ref{Lemma_vi_case2} and~\ref{Lemma_vi_case3}. First, assume that $r_i>1$. By Lemma~\ref{Lemma_vi_case1}, the following is a sequence of consecutive Lyndon words:
$L_i, L_i^{r_i-1}\vv i 1, \dots,  L_i^{r_i-1}\vv i 1 \cdots \vv i{m_i}.$ Note that
$\frac{n}{2}<|L_i^{r_i-1}\vv i 1|<|L_i^{r_i-1}\vv i 1 \vv i 2|<\cdots <|L_i^{r_i-1}\vv i1 \vv i 2 \cdots \vv i {m_i}|=n$. 
Hence, since $|L_{i+1}| \mid n$, it follows that $L_i^{r_i-1}\vv i 1\cdots \vv i {m_i}=L_{i+1}$. 
Consequently, it follows that $|L_{i+1}|=n$ thus $r_{i+1}=1$, as required.

Now, consider the case where $r_i=1$ and $|\vv i 1|\mid n$. Since $|\vv i 1|\mid n$, by the first part of Lemma~\ref{Lemma_vi_case2}, $\vv i 1=L_{i+1}$. Moreover, by the second part of Lemma~\ref{Lemma_vi_case2}, $\vv i 1=\cdots =\vv i {m_i}$ and $m_i =r_{i+1}$. Therefore, $L_i^{r_i-1}\vv i 1 \cdots \vv i {m_i}= \vv i 1 \cdots \vv i {m_i}= {\vv i 1 }^{m_i} = L_{i+1}^{r_{i+1}}$, as required.

It is left to deal with the case where $r_i=1$ and $|\vv i 1|\nmid n$. By Lemma~\ref{Lemma_vi_case3}, the following is a sequence of consecutive Lyndon words:$$\text{
$L_i$, \ $\vv i1$, \  $\vv i 1 \cdots \vv i {j_0+1}$, \ $\vv i 1 \cdots \vv i {j_0+2}$, \ $\dots$ , $\vv i 1 \cdots \vv i {m_i}$}$$ where $j_0$ is the maximal integer such that $|(\vv i 1)^{j_0}|<n$. As in the former case, $\frac{n}{2}<|\vv i1 \cdots \vv i {j_0+1}|<\cdots<|\vv i 1\cdots \vv i {m_i}|=n$. Since $|L_{i+1}| \mid n$, we conclude that $L_{i+1}=\vv i 1\cdots \vv i {m_i}$. Moreover, we get that $|L_{i+1}|=n$, thus $r_{i+1}=1$ and the lemma follows.
\end{proof}

By the former lemma and by Lemmas~\ref{Lemma_vi_case1},~\ref{Lemma_vi_case2} and~\ref{Lemma_vi_case3}, we also conclude:

\begin{corollary}
\label{cor:main-cor-of-5.2.1}
Take $L_i\neq k-1$, and consider the words $\vv i1,\cdots,\vv i {m_i}$, as defined in Definition~\ref{Def_vi}. 
\begin{enumerate}
    
    \item If $\class L=L_i^{r_i-1}\vv i 1 \dots \vv i j$ is a Lyndon word where $1\leq j\leq m_i$, then $L_i<\class L\leq L_{i+1}$.
    
    \item $L_i^{r_i}\vv i 1 \cdots \vv i {m_i}=L_iL_{i+1}^{r_{i+1}}$.
    
\end{enumerate}
\label{Cor:L_i-vvv_L_i+1}
\end{corollary}

\subsubsection{Analyzing the Structure of $w$}

We are ready to present our analysis concerning the structure of an $n$-word, $w$, in order to extract information that we use to compute $\cover(w)$. First, we identify a distinguished simple case, and define:

\begin{definition}
An $n$-word, $w$, is said to be an expanded Lyndon word, if $w=L_i^{r_i}$ for some $i\leq N(n,k)$.
\end{definition}

If $w=L_i^{r_i}$ is an expanded Lyndon word, then $\cover(w)=(L_i,\varepsilon)$.  
The reader may observe that procedures ${\sf find\_root}$ and ${\sf is\_Lyndon}$, both described in subsection~\ref{SubSubSection-FTG-Algo}, can be used to decide if $w$ is an expanded Lyndon word efficiently, and to extract $L_i$ in linear time in those cases. 

But what shall we do in the general case? Namely, if $w$ is a subword of $L_{i}L_{i+1}^{r_{1+i}}$, but not a suffix of this sequence? As a first step for answering this question we invoke Corollary~\ref{Cor:L_i-vvv_L_i+1}, which establishes relationships between $w$ and the words defined in Definition~\ref{Def_vi}, as the next lemma elaborates.

\begin{lemma}
\label{Lemma:w-divide}
Let $w\neq(k-1)^p0^{n-p}$ be an $n$-word which is not an expended-Lyndon-word. If $w$ is covered by $L_{i+1}$, then $w=yL_{i}^{r_{i}-1}\vv {i} 1 \cdots \vv {i} j z$, where $y$ is a proper non-empty suffix of $L_{i}$, and $z$ is a proper prefix of $\vv {i} {j+1}$.
\end{lemma}
\begin{proof}

$w$ is a subword of $L_1\cdots L_{i}L_{i+1}^{r_{i+1}}$. By Corollary~\ref{Cor:L_i-vvv_L_i+1}, $L_1\cdots L_{i}L_{i+1}^{r_{i+1}}=L_1\cdots L_{i-1}L_{i}^{r_{i}}\vv {i} 1 \cdots \vv {i} {m_{i}}$. Since $w$ is not covered by $L_{i}$, $w$ is a subword, but not  a prefix, of $L_{i}^{r_{i}}\vv {i} 1 \cdots \vv {i} {m_{i}}$. Furthermore, since $w$ is not an expanded Lyndon word, $w$ is not a suffix of $L_{i}^{r_{i}}\vv {i} 1 \cdots \vv {i} {m_i}=L_iL_{i+1}^{r_{i+1}}$, which proves that $w$ is of the required form.
\end{proof}

From the proof of Lemma~\ref{Lemma:w-divide}, we also conclude:

\begin{corollary}
\label{cor:y-suffix}
Assume that an $n$-word $w\neq (k-1)^p0^{n-p}$, is not an expanded Lyndon word, and $w$ is covered by $L_{i+1}$. Write $w=yL_{i}^{r_{i}-1}v_{{i},1}\cdots v_{{i},j}z$ as in Lemma~\ref{Lemma:w-divide}. Hence, $\cover(w)=(L_{i+1},x)$ where $x$ is the $|y|$-suffix of $L_{i}L_{i+1}^{r_{i+1}}$.
\end{corollary}

This corollary suggests a direction for computing $\cover(w)$. Namely, finding $L_{i+1}$ and finding the $|y|$-suffix of $L_{i}L_{i+1}^{r_{i+1}}$. For extracting this data, first, we check if the subword of $w$: $\vv {i} 1\cdots \vv {i} j$ is not empty (i.e. if $j>0$). In the case  where $j=0$, it follows that $w$ is a rotation of $L_{i}^{r_{i}}$, and this fact is used by Algorithm~\ref{Algo:filling_the_gap} to find $L_{i+1}$ and the $|y|$-suffix of $L_{i}L_{i+1}^{r_{i+1}}$. In the case where $j>0$, we use Lemma~\ref{Lemma:w-divide} and Corollary~\ref{cor:main-cor-of-5.2.1} to find a Lyndon word, $\class L$, such that $L_{i}<\class L\leq L_{i+1}$. Then, $L_{i+1}$ can be found by applying Algorithm~\ref{Algo_next_n_Lyndon} on $\class L$. It is also required to compute $|y|$ in this case. To summary, we set three goals for our analysis:
\begin{enumerate}
    \item Deciding if $j=0$.
    
    \item If $j>0$, finding a Lyndon word, $\class L$, such that $L_{i}<\class L\leq L_{i+1}$. 
    
    \item If $j>0$, computing $|y|$.
\end{enumerate}

    We start with the first goal. The next lemma provides a criterion equivalent to $j=0$.

\begin{lemma}
\label{Lemma_j_0}
Assume that an $n$-word, $w\neq (k-1)^p0^{n-p}$, is not an expanded Lyndon word, and $w$ is covered by $L_{i+1}$. Write $w=yL_{i}^{r_{i}-1}v_{{i},1}\cdots v_{{i},j}z$ as in Lemma~\ref{Lemma:w-divide}. Hence, $j>0$  if and only if $ y = (k-1)^{|y|}$.
\end{lemma}
\begin{proof}
Write $\vv {i} 1=v(\sigma+1)$. Thus, $L_{i}=v\sigma(k-1)^l$ where $l=|L_{i}|-|\vv {i} 1|$. Assume that $j>0$ and observe that, since $|w|=n=|L_{i}^{r_{i}}|$, it follows that $|y|+|\vv {i} 1|\leq |L_{i}|$. Therefore, $|y|\leq l$. Hence,  since $y$ is a suffix of $L_{i}$, it follows that $y=(k-1)^{|y|}$.

For the other direction, assume that $j=0$. Hence, $w=yL_{i}^{r_{i}-1}z$ where $z$ is a proper prefix of $\vv {i} 1$. Write $\vv {i} 1=v(\sigma +1)= zu(\sigma+1)$, and note that $z$ is also a prefix of $L_{i}$. As $|w|=n=|yL_{i}^{r_{i}}z|$, we get that $L_{i}=zy$ and hence, $y=u\sigma(k-1)^l$. Since $\sigma<\sigma+1\leq k-1$, the claim follows.
\end{proof}

At first glance, the previous lemma does not seem applicable since we aim to compute $|y|$, but we have to know $|y|$ to determine if $y=(k-1)^{|y|}$. In fact, Lemma~\ref{Lemma_j_0} actually serves as an intermediate property, which is equivalent to another property that concerns the structure of $w$, and can be computed efficiently.

\begin{definition}
An $n$-word, $w$, is said to be almost-Lyndon, if $w={(k-1)}^lu$, and $u(k-1)^l$ is an expanded Lyndon word.  
\end{definition}

\begin{lemma}
\label{Lemma-almost_Lyndon_j_not_0}
Assume that an $n$-word, $w\neq (k-1)^p0^{n-p}$, is not an expanded Lyndon word, and $w$ is covered by $L_{i+1}$. Write $w=yL_{i}^{r_{i}-1}v_{{i},1}\cdots v_{{i},j}z$ as in Lemma~\ref{Lemma:w-divide} (possibly, $r_i=1$ and then $L_i^{r_i-1}=\varepsilon$). Thus, $w$ is almost-Lyndon if and only if $j > 0$.
\end{lemma}
\begin{proof}

Before we prove the lemma, we mention a simple claim whose strait-forward proof can be found in~\cite{amram2019efficient} (Lemma 6).
\begin{quote}

Claim. Let $u\sigma(k-1)^l$ be an expanded Lyndon word. Then, $u(k-1)^{l+1}$ is an expanded Lyndon word. 
\end{quote}

We turn now to prove the lemma. First, assume that $j>0$ and hence, by Lemma~\ref{Lemma_j_0}, $y=(k-1)^{|y|}$. 
By Lemma~\ref{Lemma:L_i-vvv_L_i+1},   $L_{i}^{r_{i}-1}\vv {i} 1 \cdots \vv {i} j \cdots \vv {i} {m_i}=L_{i+1}^{r_{i+1}}$. 
Since $z$ is a prefix of $\vv {i} {j+1}$, for some $x$ with $|x|=|y|$, it holds that $L_{i+1}^{r_{i+1}}=L_{i}^{r_{i}-1}\vv {i} 1\cdots \vv {i} j zx$. 
As $L_{i+1}^{r_{i+1}}$ is an expanded Lyndon word, by applying the mentioned claim $|y|$ times, we get that $L_{i}^{r_{i}-1}\vv {i} 1 \cdots \vv {i} jz(k-1)^{|y|}=L_{i}^{r_{i}-1}\vv {i} 1 \cdots \vv {i} jzy$ is also an expanded Lyndon word. Therefore, $w=yL_{i}^{r_{i}-1}\vv {i} 1 \cdots \vv {i} j z=(k-1)^{|y|}L_{i}^{r_{i}-1}\vv {i} 1 \cdots \vv {i} j z$ is almost-Lyndon.

For proving the other direction of the equivalence assume that $j=0$, and we shall prove that $w$ is not almost-Lyndon. Since $j=0$, $w=yL_{i}^{r_{i}-1}z$ where $L_{i}=zy$. Moreover, since $w$ is not covered by $L_{i}$, $z\neq \varepsilon$. By Lemma~\ref{Lemma_j_0}, $y \neq (k-1)^{|y|}$ thus $y=(k-1)^ly'$ where $y'\neq \varepsilon$ does not start with $k-1$. Therefore, $L_{i}^{r_{i}}=L_{i}^{r_{i}-1}z(k-1)^ly'$, and we claim that $L_{i}^{r_{i}}<y'L_{i}^{r_{i}-1}z(k-1)^l$  which proves that the latter is not an expanded Lyndon word. Indeed, this inequality holds since $y'$ is a proper suffix of $L_{i}$ and $L_{i}$ is a Lyndon word. This guarantees that $y'$ is strictly larger than the prefix of $L_{i}$ of length $|y'|$. Clearly, an expanded Lyndon word cannot be strictly smaller than one of its non-trivial rotations thus, as said, this inequality proves that $y'L_{i}^{r_{i}-1}z(k-1)^l$ is not an expanded Lyndon word. Therefore, $w=(k-1)^ly'L_{i}^{r_{i}-1}z$ is not almost-Lyndon.
\end{proof}

So far, we identified a structural property of $w$ which testifies if $j=0$ or not. We turn now to achieve our second and third goals, which are finding a Lyndon word $\class L$ such that $L_{i}<\class L\leq L_{i+1}$, and computing $|y|$, when $j>0$. For these purposes, we use the classic result by Chen, Fox and Lyndon~\cite{CFL}. The authors of~\cite{CFL} (see also~\cite{duval1983factorizing}) proved that every non-empty word, $w$, can be uniquely factorized into Lyndon words: $w=x_1|x_2|\dots |x_l$ such that $x_1\geq x_2\geq \dots \geq x_l$. We name this decomposition: the $\CFL$-factorization of $w$. In the next lemma we show the connection between the \CFL-factorizing of an $n$-word, $w$, and the structure of $w$ as characterized in Lemma~\ref{Lemma:w-divide}.

\begin{lemma}
\label{Lemma_CFL}
Assume that an $n$-word,  $w\neq (k-1)^p0^{n-p}$ is not an expanded Lyndon word, and $w$ is covered by $L_{i+1}$. Write $w=yL_{i}^{r_{i}-1}v_{{i},1}\cdots v_{{i},j}z$ as in Lemma~\ref{Lemma:w-divide}, and assume also that $j>0$. Let $y=y_1|\cdots |y_l$ be the \CFL-factorization of $y$ and let $z=z_1|\cdots |z_m$ be the \CFL-factorization of $z$.

\begin{enumerate}
    
    \item If $r_{i}=1$ and $\vv {i} 1=\cdots = \vv {i} j$, then the \CFL-factorization of $w$ is $$w=y_1|\cdots|y_l|v_{{i},1}|\cdots |v_{{i},j}|z_1|\cdots|z_m$$
    
    \item Otherwise, the \CFL-factorization of $w$ is $$w=y_1|\cdots|y_l|L_{i}^{r_{i}-1}v_{{i},1}\cdots v_{{i},j}|z_1|\cdots|z_m$$

\end{enumerate}
\end{lemma}

\begin{proof} 
 The two statements are proved by similar arguments thus we prove only the second claim. We leave for the reader to observe that $L_{i}^{r_{i}-1}\vv {i} 1 \cdots \vv {i} j$ is indeed a Lyndon word. When $r_{i}>1$, this is implied by Lemma~\ref{Lemma_vi_case1}, and when $r_{i}=1$, this follows from Lemmas~\ref{Lemma_vi_case2} and~\ref{Lemma_vi_case3} since in this case, $\neg(\vv {i} 1=\cdots =\vv {i} j)$.
 
 Therefore,  $y_1|\cdots|y_l|L_{i}^{r_{i}-1}\vv {i} 1 \cdots \vv {i} j |z_1|\cdots|z_m $ is a factorization of $w$ into Lyndon words. Since $y_1|\cdots|y_l$ and $z_1|\cdots|z_m$ are the $\cfl$-factorizations of $y$ and $z$, respectively, it is left to prove:
\begin{itemize}
    \item[(a)] $y_l\geq L_{i}^{r_{i}-1}\vv {i} 1\cdots \vv {i} j$.
    \item[(b)]  $L_{i}^{r_{i}-1}\vv {i} 1\cdots \vv {i} j\geq z_1$.
\end{itemize}
Since $j>0$, $y=(k-1)^{|y|}$ thus $y_l=k-1$, which proves (a). Now, (b) holds since $z_1$ is prefix of $L_{i}$. Indeed, if $r_{i}>1$, then $L_{i}^{r_{i}-1}\vv {i} 1 \cdots \vv {i} j> z_1$ since $z_1$ is a prefix of $L_{i}^{r_{i}-1}\vv {i} 1 \cdots \vv {i} j$. In addition, when $r_{i}=1$, since $z_1$ is a prefix of $L_{i}$, by the definition of $\vv {i} 1$, $\vv {i} 1>L_{i}>z_1$ thus $\vv {i} 1 \cdots \vv {i} j>z_1$, as required.
\end{proof}

By the previous lemma, by  Lemma~\ref{Lemma_j_0} and by Corollary~\ref{cor:main-cor-of-5.2.1}, we conclude the following consequence, which achieves the two remaining goals.

\begin{corollary}
\label{cor:main-cor-of-5.2.2}
Assume that an $n$-word,  $w\neq (k-1)^p0^{n-p}$, is not an expanded Lyndon word, and that $w$ is covered by $L_{i+1}$. Write $w=yL_{i}^{r_{i}-1}v_{{i},1}\cdots v_{{i},j}z$ as in Lemma~\ref{Lemma:w-divide}, assume that $j>0$, and let $x_1|\cdots|x_l$ be the $\CFL$-factorization of $w$. If $x_t$ is the first word in this factorization, different from $k-1$, then:
\begin{enumerate}
    \item $L_{i}<x_t\leq L_{i+1}$.
    \item $|y|=t-1$.
\end{enumerate}
\end{corollary}

\subsubsection{A Linear Time $FTG$ Algorithm}
\label{SubSubSection-FTG-Algo}

We are finally ready to present our linear time $FTG$ algorithm. In addition to the algorithms described earlier, we use the following procedures, all can be computed in linear time:
\begin{description}
    
    \item[${\sf find\_root}(v)$.]  The root of a non-empty word $v$ is its shortest non-empty prefix $x$ such that  $v=x^t$ for some integer $t$. The word $x$ can be found by searching the first occurrence of $v$ within $v'v$, where $v'$ is the $(|v|-1)$-suffix of $v$. This can be done, for example, by invoking the KMP-algorithm~\cite{knuth1977fast} which runs in linear time. See~\cite{gawrychowski2008finding} for a presentation of this technique. Extensions and a detailed discussion can be found in~\cite[Chapter~8]{lothaire2005applied}.

    \item[${\sf find\_suffix}(u,v)$.] This procedure receives two words, $u,v$, where $u$ is a subword of $v$. The procedure returns a word $x$ such that $ux$ is a suffix of $v$. This procedure can be implemented by modifying the KMP-algorithm. The procedure runs in $O(|u|+|v|)$ time. We apply this procedure only on inputs of size at most $3n$ thus we refer to this procedure as a linear time procedure.
    
    \item[${\sf find\_min\_rot}(w)$.] Returns the lexicographically minimal rotation of $w$. Can be implemented by invoking Booth's algorithm~\cite{booth1980lexicographically}, or Shiloach's Algorithm~\cite{shiloach81}, both runs in linear time. 
    
    \item[${\sf is\_Lyndon}(L)$.] Tests if $L$ is a Lyndon word. $L$ is a Lyndon word if and only if it is equal to its root (tested by comparing $L$ with ${\sf find\_root}(L)$) and it is equal to its minimal rotation (tested by comparing $L$ with ${\sf find\_min\_rot}(L)$). 
    
    \item[${\sf CFL}(w)$.] Returns the \CFL-factorization of $w$. See \cite{duval1983factorizing} for details and runtime analysis.
    
    \item[${\sf is\_almost\_Lyndon}(w)$.] 
    
    Given an $n$-word, $w$, this procedure checks if it is almost-Lyndon. If $w=(k-1)^l u$, where $u$ does not start with $k-1$, we test if $w'=u(k-1)^l$ is an expanded Lyndon word. This occurs if and only if ${\sf is\_Lyndon}({\sf find\_root}(w'))$ holds.
\end{description}

\begin{algorithm}[!ht]
\caption{{\sf filling\_the\_gap}}
\textit{Input:} a word $w$ of length $n$ \\
\textit{Output:} $\ftg(w)$

\hrulefill

\begin{algorithmic}[1]

\If {$w=(k-1)^n$} \hfill // 1-7: edge cases, $w=L_i^{r_i}$ 
\State return $(k-1,\varepsilon)$\EndIf

\State $L\gets{\sf find\_root}(w)$

\If {${\sf is\_Lyndon}(L)$} 
\State return $\ctf(L,\varepsilon)$ 

\EndIf\vspace{2mm}

\If {${\sf is\_almost\_Lyndon}(w)$} \hfill // 8-10: test for $j>0$ 
\State \textbf{goto} \ref{mark}\EndIf\vspace{2mm}

\State $u\gets{\sf find\_min\_rot}(w)$ \hfill // 11-16: if $j=0$

\State $L\gets {\sf find\_root}(u)$

\State $L'\gets \lnext(L)$

\State $r'\gets n/|L'|$.

\State $x\gets {\sf find\_suffix}(w,LL'^{r'})$

\State return $\ctf(L',x)$\vspace{2mm}

\State\label{mark} $x_1|\cdots|x_l\gets {\sf CFL}(w)$ \hfill // 17-25: if $j>0$

\State $t\gets$ first integer such that $x_t\neq k-1$ 

\State $L\gets x_t$

\If {$|L|\nmid n$} 
    \State $L\gets \lnext(L)$ \EndIf

\State $r\gets n/|L|$

\State $ x \leftarrow (t-1)$-suffix of $L^r$ 
\State return $\ctf(L,x)$

\end{algorithmic}

\label{Algo:filling_the_gap}
\end{algorithm}

\begin{proposition}
\label{Prop:main-prop-of-sec5}
Algorithm~\ref{Algo:filling_the_gap} computes $FTG(w)$ in $O(n)$ time.
\end{proposition}

\begin{proof}
The algorithm clearly terminates after $O(n)$ steps, and we prove its correctness.
First, use Corollary~\ref{cor:main-cor-of-5.1} to observe that lines 1-7 handle correctly the case where $w$ is an expanded Lyndon word (use the second item of Corollary~\ref{cor:main-cor-of-5.1} for the case where $w=0^n$). Furthermore, if $w=(k-1)^p0^{n-p}$, where $0<p < n$, then $w$ is almost-Lyndon, and the algorithm terminates in line 25 and returns $\ctf(0,0^p)$. Again, item 2 of Corollary \ref{cor:main-cor-of-5.1} shows that a correct value is returned by the algorithm in these cases.

It remains to deal with the general case, in which $w\neq (k-1)^p0^{n-p}$ and $w$ is not an expanded Lyndon word.
Take such $w$, covered by $L_{i+1}$.
Write $w=yL_i^{r_i-1}v_{i,1}\cdots v_{i,j}z$  as in Lemma~\ref{Lemma:w-divide}, and assume, first, that $j=0$. Hence, by Lemma~\ref{Lemma-almost_Lyndon_j_not_0}, the test in line 8 returns a negative response, and the computation proceeds to line 11. Now, since $j=0$, $w$ is a rotation of $L_i^{r_i}$. As a result, in line 11, the word $L_i^{r_i}$ is assigned to variable $u$, in line 12, $L$ is assigned with $L_i$, and in line 13, $L'$ is assigned with $L_{i+1}$. Since $w$ is a subword of $L_iL_{i+1}^{r_{i+1}}$, line 15 assigns to $x$ the value that satisfies: $\cover(w)=(L_{i+1},x)$. Therefore, by Corollary~\ref{cor:main-cor-of-5.1}, the invocation of the $\ctf$ procedure in line 16 returns $\ftg(w)$.

Now, consider the case where $j>0$. By Lemma~\ref{Lemma-almost_Lyndon_j_not_0}, the test in line 8 returns true, and the computation traverses to line 17. By Corollary~\ref{cor:main-cor-of-5.2.2}, in line 19 a Lyndon word is assigned to variable $L$ such that $L_i<L\leq L_{i+1}$. Thus, when the computation reaches  line 23,  $L$ stores $L_{i+1}$. Moreover, by the same corollary, in line 24 $x$ is assigned with the $|y|$-suffix of $L_{i+1}^{r_{i+1}}$. Therefore, by Corollary~\ref{cor:y-suffix},  $\cover(w)=(L_{i+1},x)$ and hence, by Corollary~\ref{cor:main-cor-of-5.1}, the method invocation in line 25 returns $\ftg(w)$.
\end{proof}

By Propositions~\ref{prop:main-prop-of-sec4} and~\ref{Prop:main-prop-of-sec5} we conclude:
\begin{theorem}
Algorithm~\ref{algo_generalized-shift-rule} forms a generalized-shift-rule for the prefer-min DB sequence that runs in $O(n+c)$ time.
\end{theorem}

\section{Conclusion}
\label{Sec:Discussion}

We proposed the notion of a generalized-shift-rule for a De Bruijn sequence which, unlike a shift-rule, allows to construct the entire De Bruijn sequence efficiently. 
We noted that a generalized-shift-rule for an $(n,k)$-DB sequence runs in time $\Omega(n+c)$, and presented an $O(n+c)$ time generalized-shift-rule for the well-known prefer-min De Bruijn sequence. By imposing a trivial reduction, as explained in the preliminaries section, our results provide a generalized-shift-rule for the prefer-max De Bruijn sequence as well.

\bibliographystyle{elsarticle-num}
\bibliography{allbib}

\end{document}